\documentclass[runningheads]{llncs}
\usepackage[T1]{fontenc}

\usepackage{mathbbol}%
\usepackage{graphicx}
\usepackage{xspace}
\newcommand{\vp}{$\mathrm{VP}$\xspace}

\newcommand{\vpd}{$\mathrm{VP(\delta)}$\xspace}

\newcommand{\rvp}{$\mathrm{R\mbox{-}VP}$\xspace}

\newcommand{\trap}[1]{{\sf{TrapIn}}-$(#1)$}

\usepackage{xcolor}

\usepackage{soul}

\usepackage{algorithm}
\usepackage{algpseudocode}

\usepackage[normalem]{ulem} 
\usepackage{xcolor}

\newcommand{\move}{{\sf Move-to-Unflagged }}
\newcommand{\movek}{{\sf Generalized-Move-to-Unflagged }}

\begin{document}
\title{Online Treasure Hunt in Vertex-Permuted Dynamic Rings}
%
%

\author{
Kamran Ayoubi\inst{1}\orcidID{0009-0008-7768-8639}
\and
Bernard Mans\inst{2}\orcidID{0000-0001-7897-2043}
\and
Lata Narayanan\inst{1}\orcidID{0000-0002-3875-0371}
}

\authorrunning{K. Ayoubi et al.}

\institute{
Department of Computer Science and Software Engineering,
Concordia University, Montreal, Canada\\
\email{kamran.ayoubi@concordia.ca, lata.narayanan@concordia.ca}
\and
School of Computing, Macquarie University, Sydney, Australia\\
\email{bernard.mans@mq.edu.au}
}

\maketitle              
\begin{abstract}
 We study the problem of treasure hunt by a group of \(k \geq 1\)
agents in {\em vertex-permuted dynamic rings} (\vp). In this model, the $n$ vertices remain on a ring but are permuted at each time step.

We first show that treasure hunt is impossible for any $k \leq n-3$ agents, if there are no restrictions on the sequence of permutations used in the dynamic ring.  We then study the \vpd~setting, in which for every pair $i, j$ of vertices,  the edge $(i, j)$ is guaranteed to appear within $\delta$ steps. We  show that the class \vpd~is feasible only for $\delta \ge \left\lceil \frac{n-1}{2}\right\rceil$. For the one-agent case, we show a tight bound of $\Theta(\delta n)$ on the worst-case search time as well as competitive ratio of any online algorithm for treasure hunt, provided $\delta \ge 2n$. We then give an optimal algorithm for $k$ agents, thereby  showing that $k$ agents can obtain a speedup of $k$ on the worst-case search time. Finally, in the \rvp~setting,  in which in every step, the vertices are arranged as a ring according to a {\em random}  permutation, we show that treasure hunt takes expected $\Theta(n)$ steps against an oblivious adversary and $\Theta(n \log n)$ steps against an adaptive adversary.

\end{abstract}
\section{Introduction}
Modern networks are dynamic rather than static, with links between entities changing over time. These networks are ubiquitous and occur in critical settings, such as 
wireless and mobile networks, transportation networks, sensor networks, and systems where nodes or links can fail
or be turned off. Moreover, changes can be arbitrary, making it challenging to design algorithms for such networks.

In this paper we study the problem of {\em treasure hunt} or search by a \emph{group of agents} in a dynamic ring network. There is a special vertex $T$ (the \emph{treasure}), and $k \geq 1$ mobile agents located at potentially different vertices of the  $n$-node ring start looking for the treasure at the same time. 
Agents are identical, anonymous, memoryless, and operate synchronously and {\em online}:  at each time step $t$ they observe their neighbors in the current graph $G_t$, and then decide
whether to move to a neighboring vertex or to stay, without any knowledge of future graphs, or memory of past actions.  Their collective goal is to minimize the  time that the {\em first} agent reaches the treasure.

We focus on \textit{vertex-permuted rings} (\vp), a dynamic ring in which, in every time step, the vertex set is permuted and then arranged in a ring topology \cite{agarwalla2018deterministic,AyoubiNarayananSAND25}. 
A vertex-permuted graph is a temporal graph in which the structure of the graph remains fixed at each time step throughout the graph’s lifetime, though the vertices may assume different positions in the graph. For example,  cyclists who ride in pelotons,  swap their positions after specific time intervals as it is more tiring to  be in front of the formation. In the natural world, it has been observed that birds flying in formation or fish swimming in formation swap positions while maintaining the same structure of the formation. In the context of sensor networks that have clusters configured in a star topology, different nodes  periodically assume the role of clusterhead, in order to balance energy consumption across cluster members. All the above are examples of vertex-permuted graphs. 

We  show that treasure hunt is impossible in a vertex-permuted ring if there is no restriction on the dynamicity. In view of this impossibility, we consider restrictions on the dynamicity.
 We consider the class \vpd: for every ordered pair $(u,v)$ with $u\neq v$, and every starting time $t$, the vertices $u$ and $v$ appear as neighbors in some step
$s\in\{t,t+1,\ldots,t+\delta-1\}$. 
Finally, we consider a {\em random} version  of \vp. In \rvp, in every step, a random permutation of the vertices is arranged as a ring. 

\smallskip




\subsection{Overview of results}


We first show that for any $k$ with $1 \leq k \leq n-3$ and for every randomized or deterministic online algorithm, there are inputs on which no agent can ever reach the treasure. 

For \vpd, using an interesting connection to the Walecki decomposition of a clique into Hamiltonian cycles~\cite{alspach2008wonderful,Lucas1896}, we first note that the class is non-empty if and only if $\delta\ge \lceil(n-1)/2\rceil$. 
Next we show a tight bound of $\Theta(\delta n)$ on the search time and on the competitive ratio for a single deterministic agent in \vpd~for $\delta \geq 2n$, provided agents have the ability to {\em flag} a node when they visit it, and subsequently, agents at neighboring nodes can see the flag. In contrast, if agents do not have the ability to flag nodes, we show that no deterministic single-agent algorithm can solve treasure hunt in \vpd. We then give a randomized algorithm that can find the treasure in expected $O(\delta^2)$ time.

For $k$ agents, we show a lower bound of 
$(\lceil (n-2)/k\rceil-2)\,\delta$, provided $\delta \geq (k+1)(n-k-2)$ on the time any agent can reach the treasure; we give an algorithm that matches this bound asymptotically.

For upper bounds,  we assume that agents are anonymous and memoryless, have the ability to flag, and can see flags at distance one, but do not have chirality and do not know any node labels. Our matching lower bounds hold even if we assume agents have memory, chirality, full knowledge of node labels and visibility of the entire graph.

To our knowledge, the idea of allowing agents to plant flags at visited nodes that can be seen subsequently at distance one is novel. We note that flags constitute a compact (one-bit) but powerful form of persistent memory that obviates the need for agents to know or remember identities of nodes they have visited,  and to find out about nodes that other agents have visited either by sending/receiving messages, or via an increased visibility range.

In \rvp, using an equivalence with random walks in cliques, we show that the expected time for an agent to reach the treasure  is $\Theta(n)$ for an oblivious adversary and it i$\Theta(n \log n)$ for an adaptive adversary who knows the input sequence of random graphs.

\smallskip

\subsection{Organization}

Section 2 discusses related work, and Section 3 defines the model and the
problem. Section 4 proves impossibility in unrestricted vertex-permuted rings.
Section 5 studies the single-agent case in $VP(\delta)$ and $R$-VP. Finally, Section 6
studies the multiple-agent case in $VP(\delta)$.

\section{Related work}\label{sec:related}

\emph{Temporal graphs} (also called time-varying or dynamic or evolving graphs) model networks whose edges change over time, see, e.g., the TVG survey of Casteigts \emph{et al.}~\cite{casteigts2012time} and the distributed-computing model of
Kuhn, Lynch, and Oshman, who introduced $T$-interval connectivity~\cite{casteigts2012time,KuhnLO10}.
\emph{Dynamic rings} in particular have been studied for example, for the problem of gathering and dispersion \cite{agarwalla2018deterministic,DILUNA202079} and exploration \cite{Mandal,di2020distributed,IlcinkasW18,DiLunaDFS16}.

\emph{Vertex permutation dynamism} is a  temporal model in which every snapshot is isomorphic to a fixed
base graph, obtained by permuting vertex identities between steps. 
Vertex permuted rings were introduced in~\cite{agarwalla2018deterministic}, which studied the problem of dispersion in such graphs. 
In~\cite{AyoubiNarayananSAND25}, the authors study restless
exploration in vertex-permuted temporal graphs of arbitrary topology and show a precise characterization of
the base graphs for which restless exploration is always possible.  
The concept of exploration by restless agents (agents who must move to a neighboring node in every step) was introduced in \cite{bellitto2023restless}.

\emph{Search or treasure hunt}  has been studied in many settings, including continuous domains (e.g., the cow-path problem, first studied in \cite{Beck1964}) and discrete graphs \cite{Awerbuch1999,Bouchard2023}. 
More recently, multi-agent search has been studied, under different models of communication between the agents, different speeds of agents, various fault models, and so on.

Search in a graph \cite{Awerbuch1999,Bouchard2023,Pattanayak2024}
is similar to graph exploration; in the former, agents are looking for a specific node or object in the 
graph, while the goal of exploration is to visit every node in the graph. 
The exploration of temporal graphs, introduced by Michail and Spirakis \cite{Michail03072016}, and studied in \cite{ErlebachHK21} addresses the problem of designing temporal walks that visit all vertices of a dynamic graph. 
Most work on exploration in temporal graphs is offline, where the entire sequence of graphs is known before designing the trajectory of agents. There is relatively little work on {\em online} exploration in dynamic graphs. In \cite{di2020distributed} and \cite{Mandal}, online exploration of 1-interval connected dynamic rings was studied.
However, this is a very different model of dynamicity than ours. 


\emph{ Tokens} or {\em pebbles} are a common way to strengthen the extremely limited capacity of the agents for  exploration tasks, by allowing a small amount of persistent
state in the environment~\cite{DobrevEtAl2013Tokens}.
Generally speaking, agents can place pebbles at a node, and are only visible to agents when they are actually at the node. Pebbles were used in the context of treasure hunt in \cite{Gorain2022PebbleGO,Disser2018,Das2024CollisionfreeEB}. 
The flags that we use in our paper are similar to pebbles, but differ in that we assume the flag that can be planted by an agent when visiting a node is  visible to agents at neighboring nodes in all subsequent steps. They also differ from the model of mobile agents with lights \cite{Das2023} in which it is the robots that are equipped with lights that enable some communication with  other robots, whereas flags are left behind by robots at nodes.

When the dynamics are random rather than adversarial, agent trajectories are often analyzed using random-walk tools (hitting times, cover times) on temporal graphs.  A representative example is Avin, Kouck\'y, and Lotker,
who relate the evolution rate of the graph to cover-time behavior~\cite{AvinKouckyLotker2018CoverMixing}.

\section{Model}\label{sec:model} 
\noindent{\bf Dynamic graph model.} For our purposes, a  {\em dynamic graph} or a {\em temporal graph}  is  a sequence of graphs $G_0 = (V, E_0), G_1 = (V, E_1), G_2 = (V, E_2),$ $ \dots, G_L = (V, E_L)$, where each $G_i = (V, E_i)$ is a static undirected graph, and all $G_i$ have the same vertex set $V$, but possibly different edge sets. The number $L$ is called the \textit{lifetime} of the temporal graph.
Each graph $G_i$ corresponds to the state of the temporal graph at time step $i$, with an edge $e \in E_i$ representing a connection between two vertices that is available at time $i$. The sequence of graphs models the evolution of a network over time \cite{ErlebachHK21}. 

In this paper, we study  {\em vertex-permuted rings (\vp)} with $n$ nodes in which each snapshot $G_t$ is a  cycle
on $V$, obtained by an arbitrary permutation of the nodes arranged as a ring. Nodes are anonymous, that is, they may have labels, but  they are unknown to the agents. In each time step $t$, in the graph $G_t$, every node $v$ is connected to its two neighbors in the ring via distinctly labelled ports; the labelling of the ports is arbitrary and may not provide a globally consistent orientation.

For  $\delta\in\mathbb{N}$, the class  $\mathrm{VP}(\delta)$ consists of all dynamic graphs in \vp~ in which, for every pair
$\{u,v \}$ with $u\neq v$ and every starting time $t = i\delta +1 $ for $i \in \mathbb{N}$ there exists
$s\in\{t,t+1,\ldots,t+\delta-1\}$ such that $u$ and $v$ are adjacent in $G_s$.
Equivalently, every vertex sees every other vertex as a neighbor at least once in every window of
length $\delta$. A dynamic graph in \vpd~does not need to be a periodic graph with period $\delta$. So a graph in \vpd~does not guarantee adjacencies with all other vertices in a {\em sliding window} of length $\delta$. However, it is easy to see that for every graph in \vpd, adjacencies are guaranteed with all other vertices in a sliding window of length $2 \delta$. Thus, in the rest of the paper, we do not talk about sliding windows, and instead when speaking of \vpd, consider a graph where starting at time 1, time is divided into windows of $\delta$ steps, and in each window, each vertex is guaranteed to be adjacent to every other vertex in at least one step during the window. 



\noindent{\bf Agents and actions. }
There are $k\ge 1$ agents.  Agents are identical and anonmyous: they do not have identities and execute the same algorithm.  Agents are memoryless in the sense that they do not remember actions or information from the past in the current step. 

One of the nodes of the graph contains a treasure; we assume this treasure can only be seen when an agent is located at the same node as the treasure. To find the treasure, agents have to visit the nodes of the graph. In the $k \geq 2$ case, agents work in synchronous discrete time steps, using the well-known Look-Compute-Move model \cite{di2020distributed}.




We consider two models of visibility and what an agent can do when it visits a node. In the first model, when an agent visits a node, it is able to mark it as visited by planting a {\em flag} there (we say the agent flags the node), which is visible to any agent at a neighboring node. A node that has not been visited by any agent remains {\em unflagged}. Note that agents have no visibility about nodes at a distance more than one in any step. 
In the second model, an agent cannot plant a flag, and cannot distinguish between nodes that were previously visited from nodes that have never been visited, even when those nodes are adjacent.

Agents are {\em silent} - they do not communicate with agents at other nodes by sending messages. Planting flags to mark nodes as visited can be seen as a form of indirect communication between agents. There is no other communication between agents.

Each agent has a consistent private orientation of the ring, which designates each port as $cw$ or $ccw$. As in ~\cite{DILUNA202079}, access to the ports is by mutual exclusion: if multiple agents request the same port, exactly only of them succeeds. Requests to the two ports are arbitrated independently. 

Our lower bounds are valid for stronger models of agents; the specific ways in which agents can be strengthened are mentioned in the theorem statements. 

\vspace*{0.1in}
\noindent{\bf The treasure hunt problem.}
An instance of the treasure hunt problem is a dynamic graph $G_0, G_1, \ldots, G_L$, a set of $k \geq 1$ agents located at the same vertex in $G_0$, and a treasure located at a fixed vertex $T$ unknown to the agents. Each agent chooses its action at time $t$ based only on what it sees in its neighborhood at time $t$ and without any knowledge of $G_{t+1}$ to $G_L$. 
Given an instance of the treasure hunt problem, the goal is to minimize  $\tau_1^k$,  the first time any agent of $k$ agents reaches $T$.

We are also interested in analyzing the \emph{competitive ratio} of our online algorithms. The competitive ratio of an online algorithm $ALG$ is the worst-case ratio (over all possible input sequences $G_1,\ldots,  G_L$) of the cost of $ALG$ to the cost of an optimal offline algorithm that knows the entire input sequence in advance. Although we describe our lower bound arguments in the classic request-answer game framework, where the adversary chooses the next graph $G_t$ based on the actions of the agent so far, it is known that the lower bound also holds against an oblivious adversary for deterministic algorithms, and  against an adaptive offline adversary for randomized algorithms \cite{Borodin05}.

\section{Treasure hunt in unrestricted VP}\label{vp~with no restrictions}

In this section, we show that treasure hunt is impossible for any number of agents in \vp~when there are no restrictions on the permutations in every step.

\begin{figure}[t]
  \centering
  \includegraphics[width=0.9\linewidth]{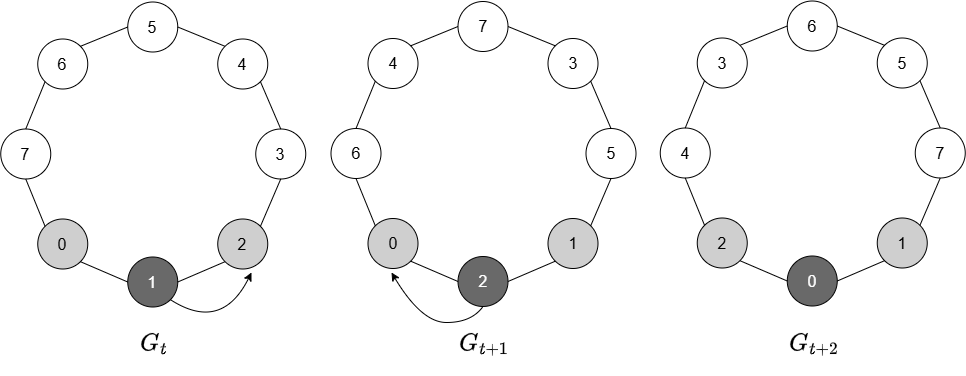}
  \caption{Illustration of the adversarial procedure \trap{3} on a vertex-permuted ring. 
  The dark vertex is the agent's current position; the other two shown vertices are its neighbors in the
  current snapshot.}
  \label{fig:trapin3}
\end{figure}

An adversarial scheduler can confine the agent to at most three nodes by using the following
\trap{3} procedure. Suppose the agent is at a node $u$ whose neighbors are $v$ and $w$ in $G_t$.
The graph $G_{t+1}$ chosen by the adversary depends on the agent's action in step $t$.
If the agent stays in $u$, then $G_{t+1} = G_t$.
If the agent moves to $v$, then in $G_{t+1}$ the adversary swaps the positions of $v$ and $u$, so
that $v$ is now in the middle. Similarly, if it moves to $w$, the adversary swaps the positions of
$u$ and $w$ so that $w$ is in the middle. In this way, the agent can never visit any nodes except
for $u, v,$ and $w$. Figure~\ref{fig:trapin3} gives a concrete example of \trap{3} with $\{u,v,w\}=\{0,1,2\}$.

\begin{figure}[t]
  \centering
  \includegraphics[width=0.95\linewidth]{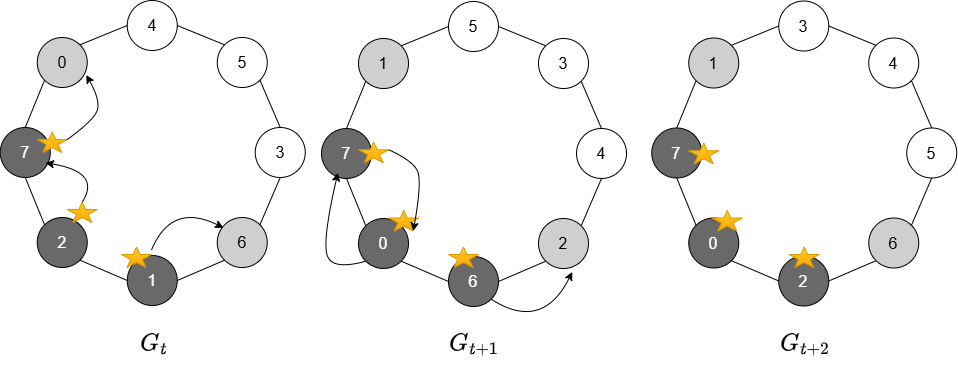}
  \caption{Illustration of the adversarial procedure \trap{k+2} for $k$ agents in a vertex-permuted ring.
Stars denote agents. The adversary keeps all agents inside a
contiguous block $A$ of $k{+}2$ vertices; whenever an agent reaches a boundary vertex of $A$, it swaps that
boundary with a free vertex of $A$ in the next snapshot, restoring confinement. In $G_t$, $a_t=0$ and $b_t =6$}
  \label{fig:trapin-k2}
\end{figure}

\begin{theorem}\label{thm:VP-delta-k-agents}
For any randomized or deterministic algorithm with $k$ agents starting at the same node, where $1 \leq k \leq  n-3$, treasure hunt is impossible, even if agents have memory, chirality, full visibility of the graph, and know the node labels.
\end{theorem}

\begin{proof}
For one agent, the \trap{3} procedure foils any deterministic or randomized agent from reaching the treasure.
The \trap{3} procedure can be generalized to any $k \leq n-3$ agents; we call the procedure \trap{k+2}
(see Figure~\ref{fig:trapin-k2}).
Given an online algorithm for treasure hunt, the adversary keeps presenting the same permutation until a
time $t$ when the $k$ agents are spread $k+2$ apart.
Let $a_t$ be the farthest node from $0$ going clockwise that contains an agent and let $b_t$ be the
farthest such node going counterclockwise, and let $A$ be the set of $k+2$ nodes between $a_t$ and $b_t$.
There is at least one node in the ring that is not in $A$; the adversary places the treasure in such a node.

Also, there must be two nodes $c_t, d_t \in A$ that are not occupied by any agents.
In time $t+1$, the adversary exchanges the positions of $c_t$ and $d_t$ with $a_t$ and $b_t$ respectively.
The agents are now in the ``interior'' nodes of $A$ and cannot reach the treasure in step $t+1$.

In every subsequent step, as long as the agents stay within the interior nodes of $A$, the adversary does
not change the permutation. If an agent moves to a ``boundary'' node, the adversary swaps the position of
the agent-containing boundary node with an interior node that does not have an agent in the next step.
Thus the agents are trapped in the set $A$ of nodes and can never reach the treasure; knowing and remembering labels of nodes, chirality, and having complete visibility does not help. 
\end{proof}

It is straightforward to see that a similar procedure can be employed even if $k$ agents do not start at the same node, provided that $n$ is large enough compared to $k$, by trapping subsets of agents in separate sets of nodes; we omit the details. 

\section{Treasure hunt by a single agent in restricted models of \vp}
 
 In light of the impossibility of treasure hunt in unrestricted \vp, we consider restricted versions of \vp in this section, namely \vpd and \rvp. 

\subsection{The \vpd model}

We first identify the minimum value of $\delta$ for which the class \vpd is
non-empty. For valid values of $\delta$, we  give a single-agent algorithm for the model when agents can flag vertices when visiting them, and then prove a matching lower bound. For the model with no flags, we show that on the one hand, no deterministic agent can perform treasure hunt, and on the other hand, there exists a randomized algorithm that completes treasure hunt in expected $O(\delta^2)$ time.

\subsubsection{Feasibility threshold for \vpd~ on a ring}
\label{sec:feasibility}

We start by considering for which values of $\delta$ the class \vpd~ is non-empty.
\begin{proposition}\label{lem:delta-threshold}
 $\mathbf{VP}(\delta)$ is non-empty if and only if
$
\delta \;\ge\; \Big\lceil\frac{n-1}{2}\Big\rceil.
$
\end{proposition}
\begin{proof}
Fix a vertex $u$. In any $\delta$-window, $u$ has exactly $2\delta$ (ordered) neighbor ``slots.'' To see all
$n-1$ other vertices at least once as a neighbor in the window, we need $2\delta\ge n-1$ i.e., the stated
lower bound. For odd $n$, a Walecki decomposition~\cite{alspach2008wonderful,Lucas1896}  of an $n$-vertex clique gives $(n-1)/2$ edge-disjoint Hamilton cycles whose union covers all edges in the clique. Each of these Hamilton cycles is a vertex-permuted ring. The sequence of all cycles in the decomposition constitutes a sequence of rings in which every node sees every other node as a neighbor. Repeating this sequence  ensures $\mathbf{VP}(\delta)$ with
$\delta=(n-1)/2$. For even $n$, a standard modification (Hamilton cycles plus antipodal pairs folded into a
cycle) yields $\delta=n/2$. Hence the bound is tight.
\end{proof}

\begin{figure}[!htbp]
  \centering
  \includegraphics[width=0.95\linewidth]{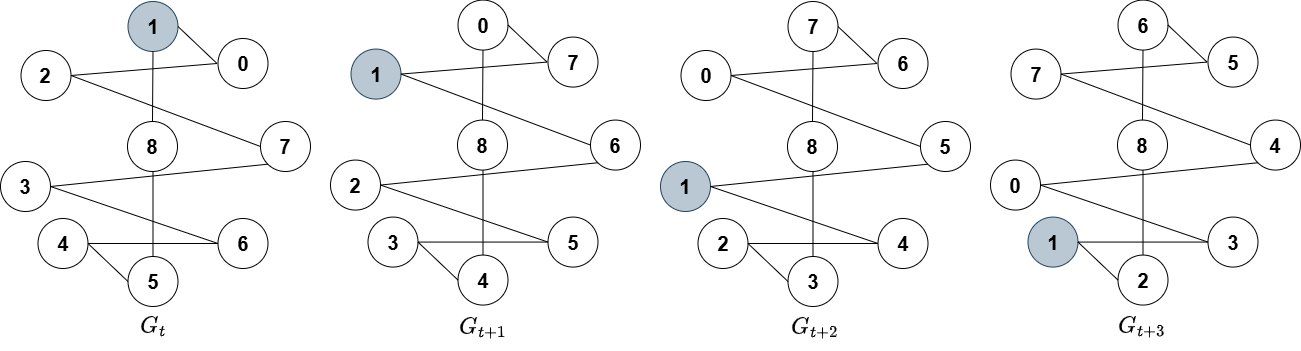}
  \caption{Example of a Walecki-style decomposition for $n=9$ shown as four consecutive vertex-permuted
  rings (since $(9-1)/2=4$). Tracking the highlighted vertex $1$ across the four rings, it becomes adjacent
  to all other vertices within a window of $4$ steps.}
  \label{fig:walecki9}
\end{figure}

\noindent\textbf{Example ($n=9$).}
Figure~\ref{fig:walecki9} illustrates the  construction for $\delta=\frac{n-1}{2}$ when $n$ is odd,
using the standard ``rotation'' view of Walecki's decomposition.
Here, vertex $8$ is kept fixed as a center vertex, and the remaining vertices are arranged on a cycle.
Each subsequent Hamiltonian cycle is obtained by rotating the outer vertices by one position
(counterclockwise in the figure) while keeping the center vertex fixed; this produces the
$(n-1)/2=4$ edge-disjoint Hamilton cycles whose union is $K_9$.

In these four consecutive rings, the highlighted vertex $1$ has neighbors
$\{8,0\}$ in step~1, $\{6,7\}$ in step~2, $\{5,4\}$ in step~3, and $\{2,3\}$ in step~4.
Thus, within $\delta=4$ steps, vertex $1$ is adjacent to every other vertex.
By symmetry, the same holds for every vertex, and by repeating this sequence, we obtain a dynamic ring in $\mathbf{VP}(4)$.

\subsubsection{Upper bound for one agent}

Consider Algorithm \move for a single agent: Until the treasure is found, in every step, if either of the neighbors is unflagged, then move to it (pick an arbitrary one if both neighbors are unflagged), otherwise  wait.

\begin{theorem}
\label{thm:vpd-single}
Algorithm \move solves the treasure hunt problem in at most $\delta(n-1)$ steps.   
\end{theorem}
\begin{proof}

First, observe that the agent never has to wait longer than $\delta$ steps to see an unflagged vertex as a neighbor, and every time it moves, it flags a new node. It follows that it completes the search within $\delta(n-1)$ steps. 
\end{proof}

\subsubsection{Lower bound for one agent}

We show that for any online deterministic algorithm, the adversary can construct an input such that the algorithm takes $\Theta(\delta n)$ steps, provided $\delta \geq 2n$, thus showing that the simple \move algorithm above is optimal for this range of $\delta$.

The adversary waits until 3 vertices have been visited and flagged. Without loss of generality, let these be vertices $\{0, 1, 2 \}$. We call these the {\em interior vertices}, and assume that they have been visited at time step 0. 
All other vertices $\{3,\dots,n-1\}$ are the \emph{exterior} vertices, which are as yet unflagged.
Time is divided into periods of length $\delta=2n$: the $r$-th period is
$r\delta, r\delta+1, \dots, (r+1)\delta-1$.
For convenience, we speak of three \emph{phases} inside a period.
The first two phases use at most $\lceil n/2\rceil$ time steps, and the three phases making up a period together use at most $2n$ time steps. The following lemma presents the properties of the input in a period.

\begin{lemma}
\label{lem:lb-vpd-single}
In every period of length $2n$, the adversary can present a sequence of vertex-permuted rings in which
(i) every pair of vertices is adjacent at least once, and
(ii) the agent can visit and flag at most one new vertex, even if it has memory, full visibility and knowledge of node labels.
\end{lemma}

\begin{proof}
Let $m=n-3$ be the number of exterior vertices. In every step but one in the period, the adversary uses the \trap{3} procedure to keep the agent in the middle node of the three interior vertices. The other two interior vertices are called the two \emph{boundary} interiors. 
The adversary chooses exactly one time step in the period to offer an unflagged vertex $v$ as a neighbor of the agent; if the agent moves to $v$, the adversary puts $v$ (and the agent) between two visited interiors in the next time step, and does not offer a second unflagged vertex as a neighbor in the same period. The challenge for the adversary is to do this while making sure that all vertex pairs are adjacent at least once during the period, thus guaranteeing that the graph is in \vpd. 

\smallskip
\noindent\textit{Phase 1 ($\lfloor n/2\rfloor -1$ time steps).}
Treat the three interior vertices as a single \emph{mega} vertex $V_m$ and form a ring on $n-2$ vertices
($V_m$ plus the $m$ exterior vertices), and use a Walecki construction to create vertex-permuted rings of $n-2$ vertices for  the next 
$\lceil (n-3)/2\rceil = \lfloor n/2\rfloor -1$
time steps. Notice that every exterior vertex will be adjacent to every other exterior vertex at least once.
In the same time, in each time step $V_m$ is adjacent to two \emph{new} exterior vertices. We now describe the adaptive part of the adversary's strategy. In each such ring of size $n-2$, we expand the vertex $V_m$, and respond to the agent's actions in the vertices that constituted $V_m$ by the \trap{3} procedure. Observe that  adjacencies between $V_m$ and the exterior vertices  are split across the three interior vertices  so that \emph{no} interior vertex meets more than $\lceil (n-3)/2 \rceil$
distinct exteriors by the end of phase~1. 
Thus by the end of phase~1 we have:
(a) all exterior–exterior pairs of vertices have been adjacent; and 
(b) at least $(n-3)$ distinct interior–exterior pairs have been adjacent, with no interior vertex being adjacent to more than half of the exterior vertices. 

\smallskip
\noindent\textit{Phase 2 ($\lfloor n/2\rfloor -1$ steps).}
In this phase, the adversary continues to use the {\em \trap{3}} procedure (Figure~\ref{fig:trapin3}) to keep the agent in the middle of the three interior vertices. However, in each time step, we ensure that each boundary interior vertex is adjacent to a new exterior vertex that it did not see in Phase 1. This is possible since  no interior vertex has been adjacent to  more than $(n-3)/2$ exterior vertices in Phase~1. The other adjacencies between exterior vertices are arbitrarily chosen. After   $\lceil n-3/2\rceil = \lfloor n/2\rfloor -1$ time steps, (at least) another $n-3$ interior–exterior pairs have been adjacent.  

\smallskip
\noindent\textit{Phase 3 (at most $n-3$ time steps).}
There are $3(n-3)$ pairs of interior-exterior vertices in all.  After Phase 2, all but $(n-3)$ of these pairs have already been adjacent during the current period. In Phase 3, in each step, to build the permutation, we first move the agent to the middle of the three interior vertices. Call this middle vertex $v$, and let $u$ and $w$ be the boundary interior vertices. If one of $u$ and $w$  still has not been adjacent to some exterior vertex, then we provide that adjacency (or those adjacencies) now and build the rest of the permutation in an arbitrary way. If both $u$ and $w$  have already been adjacent to all exterior vertices, all missing exterior-interior adjacencies are with the node $v$, where the agent is located. So in the new permutation, we swap the positions of $u$ and $v$, making $v$ a boundary interior, then make $v$ adjacent to one of those (possibly unflagged) exterior vertices $x$, swapping the positions of $u$ and $x$. This provides the only chance for the agent to visit an unflagged vertex in this period. If the agent does not move to $x$, we repeat the procedure described here, creating a new adjacency with $v$ and a new exterior vertex. If the agent does move to $x$, we now put $x$ in between $u$ and $w$ (that is, $x$ is now an interior vertex), and put $v$ in the ring at a position one away from $u$, and in each subsequent step, (a) create two of the remaining adjacencies between $v$ and exterior vertices;  (b) use \trap{3} to keep the agent in the middle vertex among $u, w, x$; (c) arbitrarily create the rest of the permutation. When this process completes, all vertices have been adjacent to all other vertices, and at most one vertex has been flagged in the period.
The total number of steps in all three phases, that is, the length of the period is $4 \lfloor n/2 \rfloor - 4 < 2n$.
\end{proof}

\begin{theorem} \label{thm:lb-vpd-single}
For any deterministic online algorithm, there is an input in $VP(\delta)$ for $\delta \geq 2n$ such that treasure hunt takes at least $\delta(n-3)$ steps, that is, $\tau_1^1 \geq \delta (n-3)$, even for agents with chirality, memory, full visibility and knowledge of node labels.
\end{theorem}

\begin{proof}
We start with 3 visited/flagged vertices $\{0, 1, 2\}$.
By Lemma~\ref{lem:lb-vpd-single}, each period covers all pairs and adds at most one new flagged vertex. Note that the set of 3 interior vertices can change between phases; interior vertices are always flagged, and some subset of the exterior vertices of size at most $i$ are flagged after the $i^{th}$ period.
 Placing the treasure at the last unflagged vertex to be reached ensures that $(n-3)$ periods are needed.
Each period has length $\delta$. Therefore the total time needed is $\delta(n-3)$.
\end{proof}

As a consequence of Theorems~\ref{thm:vpd-single} and \ref{thm:lb-vpd-single}, we have:
 \begin{corollary}
Algorithm \move, in which agents do not have chirality, memory, or knowledge of node labels, and can only see flags at neighboring nodes, is asymptotically optimal in terms of both worst-case search time  for \vpd with $\delta \geq 2n$, even considering agents with chirality, memory, full visibility, and knowledge of node labels. 
\end{corollary}

\begin{theorem}
The competitive ratio of any deterministic online agent with chirality, memory, and full visibility, but without knowledge of node labels, is $\Omega(\delta n)$.
\end{theorem}

\begin{proof}
To see that the competitive ratio of any deterministic online agent is $\Omega(\delta n)$, we claim that on the input described in the proof of Lemma~\ref{lem:lb-vpd-single}, it is possible for an optimal offline algorithm to reach the treasure in a single step. Recall that in the construction of the bad input for the online agent, we assume that the first three vertices visited by the agent were 0, 1, and 2. If the treasure was at $n-1$, then the optimal offline algorithm could reach it in one step, while the adversary can ensure that vertex $n-1$ is the very last node visited by the online agent.
\end{proof}

\subsubsection{Agents that cannot flag visited vertices}

In the \vpd~model we have considered so far, agents can distinguish between vertices that have already been visited and those that have not been visited, by using flags.
We now consider the setting
without flags. In this setting, an agent can still observe its current neighbors
and choose whether to move or stay, but it cannot tell whether a neighboring vertex has been visited before.
We first show that without flags, no deterministic single-agent algorithm can solve treasure hunt. In contrast, we give a randomized algorithm that can solve the problem in expected $O(\delta^2)$ time, by choosing to wait with a positive probability in every step.

\begin{theorem}
\label{Theorem_No_flag_No_Solution}
For \(\delta\geq n\), no deterministic single-agent algorithm can solve treasure hunt in \(VP(\delta)\) without flags, even for agents with memory, chirality, and full visibility, but no knowledge of node labels.
\end{theorem}

\begin{proof}

Consider an arbitrary deterministic algorithm for a single agent without flags. We allow the agent to have memory, but the vertices are anonymous, as far as the agents are concerned.  In this case, any deterministic algorithm can be specified as a sequence $(i, d_i)$ with  \(d_i\in\{-1,0,1\}\),  where \(-1,0,1\) respectively stand for move counterclockwise, stay, and move clockwise. The adversary knows this sequence; we give  a  strategy for the adversary for a window of length $\delta \geq n$. The adversary keeps the agent confined to two vertices \(u\) and \(v\). The treasure is placed at an arbitrary vertex outside \(\{u,v\}\).

 We treat the pair of vertices $u, v$ as a mega-vertex $V_m$; the adversary uses a Walecki construction on a graph of $n-1$ vertices in $\left\lceil \frac{n-2}{2} \right\rceil$ time steps. This ensures that all vertices except $u$ and $v$ are adjacent to each other in some step, and also, adjacent to $V_m$ in some step. 
 
 Now we describe what happens inside $V_m$ during these steps. Suppose after step $i-1$, the agent is at $u$. Then if $d_i=0$, in the graph $G_i$, the adversary keeps the relative positions of $u$ and $v$ the same as in the previous step; if $d_i=1$, the adversary makes $v$ the clockwise neighbor of $u$; and if $d_i=-1$, the adversary makes $v$ the counterclockwise neighbor of $u$. An identical strategy is used if the agent is in $v$ after step $i-1$, exchanging the roles of $u$ and $v$. In this way, the adversary ensures that if the agent moves in step $i$, it can only move from $u$ to $v$ or from $v$ to $u$. 

 At the same time, the permutation built by the Walecki construction ensures that in every step,
 at least  $\lfloor \frac{n-2}{2} \rfloor$ of the  remaining $n-2$ vertices have been adjacent to $u$ (but not to $v$) and the remaining have been adjacent to $v$ (but not to $u$).
In the next $\lceil \frac{n-2}{2} \rceil$ steps, the adversary creates these missing adjacencies for $u$ and $v$,  while using an identical strategy as in the first $\lceil \frac{n-2}{2} \rceil$ steps to keep the agent in the vertices $u$ and $v$. The adjacencies between the $n-2$ vertices other than $u$ and $v$ can be arbitrarily set.

Therefore, in at most \(2\left\lceil \frac{n-2}{2}\right\rceil \le n\)
steps, every pair of vertices has appeared as an edge at least once. Repeating
this construction in every window gives a dynamic graph in \vpd~ for every
\(\delta\ge n\). However, throughout the whole execution the agent visits only
the two vertices \(u\) and \(v\). The agent's memory may affect the sequence of
actions \(d_i\), but it does not allow the agent to distinguish \(u\) and \(v\)
from previously unseen anonymous vertices. Since the treasure is placed outside
\(\{u,v\}\), the agent never finds the treasure. Hence, no deterministic online
algorithm can solve treasure hunt in \vpd~ without flags,
even if the agent has memory.
\end{proof}


In  Theorem~\ref{Theorem_No_flag_No_Solution},  the adversary uses its advance knowledge of the agent's deterministic choices in every step to keep the agent trapped between two vertices and still create a graph in \vpd.  However, for a randomized algorithm, since the adversary does not know the agent's random choices in advance, it is possible to succeed in finding the treasure. 
Consider the following simple randomized algorithm. At each time step, the
agent stays at its current vertex with probability $P_s$. Otherwise, it moves to
one of its two neighbors, choosing the clockwise and counterclockwise neighbor
with equal probabilities
$P_\ell=P_r=\frac{1}{2}(1-P_s).$

By choosing the right value of $P_s$, we obtain the following result:

\begin{theorem}\label{thm:opt-probs-d0-nomarks}
There is a randomized algorithm for treasure hunt by one agent in $VP(\delta)$ with $\delta \in \Omega(n)$, which achieves $\mathbb{E}[\tau_1^1]=O(\delta^2)$.
\end{theorem}

\begin{proof}
Consider a particular window of $\delta$ steps. Let $u$ be the location of the agent at the first time step of the window, and let $T$ be the location of the treasure. We first find the probability that the agent reaches the treasure within the window of $\delta$ time steps.  Since the input is in $\mathbf{VP}(\delta)$, the edge $(u,T)$ must appear in this window.  Let  $t^*$ be the step chosen by the adversary to present the edge $(u, T)$.   The agent reaches $T$ at time $t^*$ if it 
(a) remains at $u$ for the first $t^*-1$ steps and
(b) moves at step $t^*$ toward $T$. 

Let $p_{\mathrm{win}}(P_s)$ denote the probability that the agent reaches $T$ in a time window of $\delta$ steps. Then 
\[
p_{\mathrm{win}}(P_s)\ \ge\ (P_s)^{t^*-1} \frac{(1-P_s)}{2}\ge\ (P_s)^{\delta-1} \frac{(1-P_s)}{2}\,>\,0.
\]
This bound is robust against an adaptive offline adversary: the adversary
chooses $t^*$, but cannot make the probability above vanish when $0 < P_s < 1$.

The lower bound on the probability is maximized when $P_s=(\delta-1)/\delta$. 
Substitute into the lower bound on $p_{\mathrm{win}}(P_s)$ and use
$\big(1-\tfrac{1}{\delta}\big)^{\delta-1}\!\to e^{-1}$, we obtain 
\[
p_{\mathrm{win}}^* \ \ge\ \frac{1}{2}\Big(1-\frac{1}{\delta}\Big)^{\delta-1}\cdot\frac{1}{\delta}
\ \approx\ \frac{1}{2e\,\delta}.
\]
Consequently, the expected number of windows to reach the treasure is at most $1/p_{\mathrm{win}}^*$,
and the expected number of steps  satisfies
\[
\mathbb{E}[\tau_1^1]\ \le\ \frac{\delta}{p_{\mathrm{win}}^*}
\ \leq  2e\,\delta^2.
\]

\end{proof}



Theorem~\ref{thm:opt-probs-d0-nomarks} shows that treasure hunt can be solved in expected $O(n^2)$ time  in a dynamic graph in \vpd~(provided $\delta = \Theta(n)$) by a randomized algorithm even when the agent does not  have the ability to flag nodes while visiting.

Observe that our randomized strategy relies on a positive probability 
of waiting. If $P_s=0$, then the lower bound on $p_{\mathrm{win}}$ given above
is zero. In fact, we can show that our \trap{3} strategy can be adapted to foil a  {\em restless} randomized agent, that is, an agent that always moves in every step; we omit the details. 

\begin{remark}[Unknown \(n\) and \(\delta\)]
In this remark, we relax the memoryless assumption and allow the
agent to maintain a counter of the number of steps since the start
of its execution. The randomized algorithm above does not require knowledge of \(n\). The assumption that the agent knows \(\delta\) can also be removed by using the
standard doubling technique; see, for example, \cite{LuoSchapire2014}.
Let \(D_i=2^{i+1}\), for \(i=0,1,2,\ldots\). In phase \(i\), the agent uses
\[
 P_s=1-\frac{1}{D_i}
 \qquad\text{and}\qquad
 P_\ell=P_r=\frac{1}{2D_i},
\]
and follows this strategy for \(\lceil 16eD_i^2\rceil\) steps. If the
treasure has not been found, the agent proceeds to phase \(i+1\), thereby
doubling its current estimate of \(\delta\).

Let \(i^*\) be the first phase for which \(D_{i^*}\geq\delta\). Since every
window of \(D_{i^*}\) steps contains a window of \(\delta\) steps, an input
in \(\mathbf{VP}(\delta)\) is also an input in
\(\mathbf{VP}(D_{i^*})\). Hence, by the proof of
Theorem~\ref{thm:opt-probs-d0-nomarks}, if this phase were continued
indefinitely, its expected search time would be at most
\(2eD_{i^*}^2\). It follows from Markov's inequality that the probability
of not finding the treasure during the phase is at most \(1/8\). The same
conclusion holds for every subsequent phase.

The total number of steps spent before phase \(i^*\) is
\(O(D_{i^*}^2)\). Moreover, the lengths of the subsequent phases increase
by a factor of four, whereas the probability of reaching each subsequent
phase decreases by a factor of at least eight. Therefore, their expected
total length is also \(O(D_{i^*}^2)=O(\delta^2)\). Indeed, for every \(m\geq 1\), the probability that the agent fails in all
of the first \(m\) phases starting from phase \(i^*\) is at most $ \left(\frac{1}{8}\right)^m.
$
Letting \(m\) tend to infinity, the probability that the agent never finds
the treasure is therefore zero. Consequently, the
treasure is found with probability one and in expected \(O(\delta^2)\)
steps, even when the agent initially knows neither \(n\) nor \(\delta\).
\end{remark}

\subsection{Treasure hunt in \rvp} \label{sec:rvp}

Next, we consider treasure hunt in \rvp, a random-order arrival model, where the permutation of vertices is chosen uniformly at random at each time step. We show that the agent can complete the search significantly faster in this setting. We consider the worst-case search times under both an oblivious adversary and an adaptive adversary. The oblivious adversary must choose the location of the treasure without knowledge of the input sequence of random graphs, while the adaptive adversary has knowledge of the input sequence. We show that  a simple search algorithm in which the agent simply moves to its clockwise neighbor is stochastically equivalent to a random walk on the complete graph $K_n$. Note that the agent does not need to flag vertices. Let $\tau(G, T)$ be the search time of the algorithm where $G$ is a dynamic graph  in \rvp and $T \in V(G)$ and $T$ is chosen by the oblivious adversary, that is, before the search algorithm starts executing. Let $T(G)$ be a treasure location chosen by an adaptive adversary that has knowledge of \cal{G}. Note that since $G$ is a random graph, $T$ is a random variable.

\begin{theorem}

For every $T \in V(G)$, and $G$ in  \rvp,
  \[
        \mathbb{E}[\tau(G, T)]
        =
        \Theta(n). \]

For every $G$ in \rvp, the adversary can choose $T(G)$ such that
 \[
        \mathbb{E}[\tau(G, T(G))]
        =
        \Theta(n \log n).
    \]
    
    \end{theorem}

    \begin{proof}
Fix a time step $t$ and let the agent be at vertex $i$ at time $t$. For any $ j \neq i$, the probability that $j$ is the clockwise neighbor of $i$ in a permutation chosen uniformly at random is $\frac{1}{n-1}$ and therefore the probability that the agent moves to $j$ is $\frac{1}{n-1}$. Since this holds for any $j \neq i$, the search procedure we described is identical to a simple random walk on a clique. 

It follows that for any treasure location picked by an oblivious adversary, the expected time to reach it is at most the hitting time in a clique, and is therefore $\Theta(n)$. On the other hand, an adaptive adversary that knows the sequence of random graphs in advance can ensure that the treasure is located at the last vertex visited by the agent. Hence the expected time to reach the treasure in this case is the cover time in a clique, which is $\Theta(n \log n)$.

    \end{proof}

\section{Search by $k \geq 2$ agents}
\label{sec:vpd-multiple}

In this section, we study treasure hunt with $k \geq 2$ agents in $VP(\delta)$. We first give
an online algorithm with search time $O(n\delta/k)$, showing that $k$ agents
obtain a linear speedup up to constant factors. We then prove a matching lower
bound for the stated range of $\delta$.

\subsection{Upper bound for $k$ agents in $VP(\delta)$}

 We now present an online algorithm \movek for \(k\) agents that finds the treasure in \(O(n\delta/k)\) time steps. The \(k\) agents may initially be located at arbitrary vertices. Agents move only to unflagged neighbors. Using the fact that access to ports is by mutual exclusion \cite{di2020distributed}, we can ensure that exactly one of the agents situated at a vertex move to each of its unflagged neighbors. The pseudocode is given in Algorithm \ref{alg:port-move-unflagged}
 .

\newcommand{\IndentedState}[1]{%
  \State \hspace{\algorithmicindent}%
  \parbox[t]{\dimexpr\linewidth-\algorithmicindent\relax}{#1}%
}
\begin{algorithm}[H]
\caption{\(\textsc{Generalized-Move-to-Unflagged}\)}
\label{alg:port-move-unflagged}
\begin{algorithmic}[1]

\State \textbf{Look:}
\IndentedState{Observe the two neighboring vertices and whether each of them is flagged.}
\IndentedState{Let \(P\) be the set of incident ports leading to unflagged neighbours.}

\State \textbf{Compute:}
\IndentedState{If \(P=\emptyset\), do not request any port.}
\IndentedState{Otherwise, order the ports of \(P\) according to the agent's private local orientation.}
\IndentedState{Request the first port in this order.}
\IndentedState{If the request fails and \(P\) contains a second port, request the second port.}

\State \textbf{Move:}
\IndentedState{If either of the port requests succeeded, move through the acquired port and flag the vertex reached.}
\IndentedState{Otherwise, stay at the current vertex.}

\end{algorithmic}
\end{algorithm}

\begin{theorem}
For $\delta \geq \lceil \frac{n-1}{2} \rceil$, \movek is an  online algorithm for treasure hunt with $k$ agents in \vpd that achieves worst-case search time $\tau^k_1 \le \delta\lceil \frac{2n}{k}\rceil$.

\end{theorem}

\begin{proof}
We show that \movek achieves this bound. 
We divide time into windows of length $\delta$. Let $W_i$ be the $i^{th}$ window of $\delta$ steps. Let $n_i$ be the number of unflagged vertices at the start of window $W_i$. Observe that in \movek, an agent moves {\em only} to an unflagged node, and at most 2 agents move to the same unflagged node in the same time step  Indeed, every unflagged node has degree two in the current ring, and at most one agent can acquire each of the two ports leading to it.  Therefore, the number of new nodes flagged in a window is at least half the number of agents that moved in that window. 

We claim that at least $min \{ n_i, \lceil k/2 \rceil \}$ nodes will be flagged during window $W_i$. First consider the case when $n_i \geq \lceil k/2 \rceil$. If every agent moves at least once during $W_i$, then at least $\lceil k/2 \rceil$ nodes were flagged, and the claim is proved. Otherwise, there exists an agent $A$ that remains at the same node $u$ during the whole window $W_i$. 
Since all nodes, and in particular, all unflagged nodes, become neighbors of $u$ during $W_i$, the only reason $A$ did not move is that another agent at \(u\) acquired the corresponding port and moved to an unflagged node. In other words, all $n_i$ unflagged nodes were visited during $W_i$, and the claim holds. 

Now consider the case when $n_i< \lceil k/2 \rceil$. Then, it cannot be that all $k$ agents moved during window $W_i$, and there must be some agent $A$ that did not move. As argued above, this implies that all 
$n_i = min \{ n_i, \lceil k/2 \rceil \}$ unflagged nodes were visited by agents during $W_i$. 

To conclude the proof, we observe that the number of windows of length $\delta$ needed for all nodes to be visited is at most  $\lceil n/ \lceil k/2  \rceil  \rceil \leq \lceil 2n/k \rceil $ and therefore treasure hunt is completed in $\delta \lceil 2n  /k \rceil $ steps. 
\end{proof}


\subsection{Lower bound for $k$ agents in $VP(\delta)$}

The lower bound argument in Lemma~\ref{lem:lb-vpd-single} and Theorem~\ref{thm:lb-vpd-single} can be generalized to $k$ agents, provided $k+2 \leq \frac{n}{2}$.  We first run the online algorithm as long as it takes for the $k$ agents to be spread out over a block of $k+2$ nodes for the first time. We ignore these initial steps. Denote by $I$ (the interior nodes) the block of nodes of length $k+2$ containing the agents. Call the remaining nodes $E$ (the exterior nodes).  The treasure will be placed in a suitable node in $E$, so the agents have to spread out in a distance $k$ block at some point, otherwise they cannot find the treasure.

\begin{lemma}
\label{lem:lb-vpd-multiple}
In every period of length $(k+1)(n-k-2)$, with $1\le k \le \frac{n}{2}-2$, the adversary can present a sequence of vertex-permuted rings in which
(i) every pair of vertices is adjacent at least once, and
(ii) the $k$ agents can collectively visit and flag at most $k$ unflagged vertices, even if they have memory, chirality, full visibility and knowledge of node labels.
\end{lemma}
\begin{proof}
    
  As in the one-agent case, we build a \emph{period} with three phases. Fix $m = n- k-2$, the size of the set $E$. For simplicity, we assume that $\delta$ is exactly $(k+1)m$.

\smallskip
\noindent\textbf{Phase 1 ($\lceil m/2\rceil$ steps)}: 
Treat the whole interior block $I$ as a mega-vertex $\mathbf{V_m}$ and consider the complete graph on
$E\cup\{\mathbf{V_m}\}$, which has $m+1$ vertices. Using a Walecki Hamiltonian decomposition, we can list
$\lceil m/2\rceil$ Hamilton cycles whose union covers all edges of this clique. We realize each cycle
as a ring on $V$ by expanding $\mathbf{V_m}$ into a contiguous block containing the vertices of $I$
(in some order), while keeping the cyclic order on $E$ as prescribed by the cycle. Hence, in these
$\lceil m/2\rceil$ steps, every $E$-$E$ pair is adjacent at least once.

In parallel, during the same $\lceil m/2\rceil$ steps we permute the order of vertices inside $I$
according to a Walecki schedule on $K_{|I|}$, which is feasible because $m\ge k+2$ implies
$\lceil m/2\rceil \ge \lceil(|I|-1)/2\rceil$. Thus all $I$-$I$ pairs are also covered by the
end of Phase~1.

Finally, in each step, the mega-vertex $\mathbf{V_m}$ has two exterior neighbors, so the expansion
creates two $I$-$E$ adjacencies incident to the two boundary interiors. Therefore Phase~1
creates $2\lceil m/2\rceil\ge m$ distinct $I$-$E$ pairs. Moreover, each boundary interior is incident
to at most $\lceil m/2\rceil$ distinct exterior neighbors by the end of Phase~1.

Throughout Phase~1, the adversary uses the \trap{k+2}~procedure to ensure that both boundary interiors
are unoccupied by agent; hence agents cannot enter $E$ in Phase~1. 

\smallskip
\noindent\textbf{Phase 2 ($\lfloor m/2\rfloor$ steps): }
We continue to keep the agents away from the boundary interiors using the \trap{k+2}~ procedure.
In each step, we choose the two exterior vertices adjacent to the two boundary interiors so that each
boundary interior meets an exterior vertex it did not meet in Phase~1. This is possible because each
boundary interior met at most $\lceil m/2\rceil$ exteriors in Phase~1, so it still has at least
$m-\lceil m/2\rceil=\lfloor m/2\rfloor$ unseen exteriors.
After $\lfloor m/2\rfloor$ additional steps, we have executed a total of
$\lceil m/2\rceil+\lfloor m/2\rfloor=m$ steps in Phases~1--2.
Because of the \trap{k+2} procedure, the \emph{identity} of the two boundary
interiors may change over time; hence we cannot assert that two \emph{fixed} interior vertices have met
all exterior vertices by time~$m$. We can ensure that
 by time  $m$, we have realized
\emph{exactly $2m$ distinct} $I$-$E$ pairs, and  of the $(k+2)m$ possible $I$-$E$ pairs, exactly 
$km$ pairs remain to be realized.

\smallskip
\noindent\textbf{Phase~3 (at most $km$ steps).}
Let $I = \{u_1, u_2, \ldots, u_{k+2}\}$. We say a node in $I$ has {\em finished} once it has already been adjacent to all other nodes in this period. If there exist two nodes $u_i, u_j \in I$ that do not contain agents and are {\em not finished}, that is, they are still missing adjacencies with exterior nodes, we make them the boundary interior nodes, and place suitable exterior nodes next to them to take care of two missing adjacencies in the step. The remaining interior nodes are placed in between $u_i$ and $u_j$ in a block of $k+2$ nodes in arbitrary order, and the remaining exterior nodes make up another contiguous block between the two chosen exterior nodes in arbitrary order. 

If instead there exists one node $u_i \in I$ that does not contain agents and is still not finished, we make $u_i$ a  boundary interior node, we pick another node $u_j \in I $ without agents as the other boundary node. Such a node must exist, since there are $k$ agents and $k+2$ nodes in $I$. Furthermore $u_j$ must have finished. We place a suitable exterior node next to $u_i$ to take care of a missing adjacency. The remaining interior nodes and exterior nodes are placed arbitrarily as in the previous case.

Otherwise, the only nodes that do not contain agents have already finished. Take such a node, call it $u_i$ that is already finished and does not have agents.  
We make $u_i$ a boundary node. For the other boundary, pick a node $u_j$ from $I$ that is still missing adjacencies. The node $u_j$ contains some agents. We will now realize all missing adjacencies for $u_j$ so that $u_j$ is finished, while ensuring that at most one exterior node is flagged.
Place a suitable exterior node $x$ next to $u_j$ (that is, $x$ has not yet been adjacent to $u_j$ in this period), and place the remaining nodes interior nodes between $u_i$ and $u_j$ as in the other cases, and the remaining exterior nodes in a contiguous block.  between $x$ and $u_i$.

\begin{itemize}

\item If no agent moves from $u_j$ to $x$, the set $I$ is intact, and we repeat this procedure in the next step (once again placing an exterior node next to $u_j$, and finding an interior node with no agents for the other boundary interior node). 

\item If some agent does move to $x$ from $u_j$, we will temporarily change $I$ and follow a different {\em special procedure}. Note that if $x$ was unflagged, it will now be flagged, and we have increased the number of flagged vertices. However, this is the only time agents can visit an unflagged exterior vertex, while we follow the special procedure to finish $u_j$. We remove $u_j$ from $I$, and bring $x$ into $I$ temporarily. Until $u_j$ misses adjacencies with other nodes in $E$, we create the next few permutations by realizing the missing adjacencies for $u_j$, meanwhile using the \trap{k+2} procedure to keep all the agents in $I$. Once $u_j$ is finished, we put $x$ back in $E$, and $u_j$ back in $I$, and revert to the original procedure, described at the beginning of Phase 3.

\end{itemize}

Note that the special procedure of finishing a node from $I$ can be done only after at least two nodes from $I$ are already finished. This implies that during the entire period, the special procedure is done at most $k$ times, and at most $k$ previously unflagged vertices can be visited during the period.

By the end of Phase~3, all  pairs of adjacencies have been realized. All three phases together take $\lceil m/2 \rceil + \lfloor m/2 \rfloor + km$ steps. Therefore, every unordered pair 
$\{u,v\}\subseteq V$ appears as an edge
in at least one step within the period with $\delta =(k+1)m=(k+1)(n-k-2)$ steps, and at most $k$ new vertices are flagged, as claimed. Having chirality, knowledge of node labels, visibility of node labels, and memory does not help as agents collectively are not adjacent to more than $k$ unflagged vertices during a period.

\end{proof}

\begin{theorem}
\label{thm:lb-vpd-multiple}
For $1\le k \le \frac{n}{2}-2$ and $\delta \geq (k+1)(n-k-2)$, and for every  online algorithm $\mathcal{A}$ with $k$ agents, there exists an input in \vpd~such that
$
\tau_1^k \ \ge\ (\left\lceil \frac{n-2}{k} \right\rceil-2)\,\delta$,  even for agents with chirality, memory, full visibility and knowledge of node labels. \end{theorem}

\begin{proof}
By Lemma~\ref{lem:lb-vpd-multiple}, in every period of length $\delta$ the agents can collectively visit at
most $k$ previously unvisited vertices.
After $r$ full periods, at most $rk$ new  vertices can have been visited. Recall that $k+2$ vertices are visited before the first period starts. 
Therefore after $r^\star=\lceil (n-k-2)/k \rceil - 1 = \lceil (n-2)/k \rceil - 2$ periods, there exists at least one exterior vertex that has not
been visited yet. Place the treasure at such a vertex. This treasure cannot be reached before the end of period $r^\star$, hence $
\tau_1^k \ \ge\ r^\star \delta \ =\ (\lceil (n-2)/k\rceil-2)\,\delta$. This proves the theorem.
\end{proof}

\begin{corollary}
\movek is asymptotically optimal in terms of worst-case search time for \vpd with $\delta \geq (k+1)(n-k-2)$.
\end{corollary}

\section{Conclusion and open problems}

We studied treasure hunt in vertex-permuted dynamic rings. We showed that, in
unrestricted \(VP\), the treasure may be hidden forever from any set of
\(k\leq n-3\) deterministic or randomized agents. 
For \(VP(\delta)\) with \(\delta\geq \lceil(n-1)/2\rceil\), we gave tight bounds 
for the search time and competitive ratio for a single deterministic
agent with flag, and showed that \(k\) agents can obtain a linear speedup. We also
studied the random model \(R\)-VP.

Several questions remain open. First, our lower bound for one deterministic
agent with flags is proved for \(\delta\geq 2n\). It would be interesting to know
whether the same bound holds for the full feasible range
\(\delta\geq \lceil(n-1)/2\rceil\). Second, the competitive ratio of randomized algorithms against an adaptive online adversary is yet to be studied.  Finally, one may ask which results
extend beyond rings to vertex-permuted graphs with other topologies.

\newpage

\bibliographystyle{splncs04}
\bibliography{paper.bib}


\end{document}